\documentclass[referee,sn-basic]{sn-jnl}

\usepackage{graphicx}%
\usepackage{multirow}%
\usepackage{amsmath,amssymb,amsfonts}%
\usepackage{amsthm}%
\usepackage{mathrsfs}%
\usepackage[title]{appendix}%
\usepackage{xcolor}%
\usepackage{textcomp}%
\usepackage{manyfoot}%
\usepackage{booktabs}%
\usepackage{algorithm}%
\usepackage{algorithmicx}%
\usepackage{algpseudocode}%
\usepackage{listings}%

\usepackage{subcaption}
\usepackage{caption} 

\usepackage{enumerate} 

\usepackage{bbm} 

\theoremstyle{thmstyleone}%
\newtheorem{theorem}{Theorem}
\newtheorem{proposition}[theorem]{Proposition}%

\theoremstyle{thmstyletwo}%

\theoremstyle{thmstylethree}%

\begin{document}

\title[Article Title]{Large-Scale Correlation Screening under Dependence for Brain Functional Connectivity Network Inference}








\author*[1]{\fnm{Hanâ} \sur{Lbath} \email{hana.lbath@inria.fr}}

\author[2]{\fnm{Alexander} \sur{Petersen} \email{petersen@stat.byu.edu} }

\author[1]{\fnm{Sophie} \sur{Achard} \email{sophie.achard@univ-grenoble-alpes.fr} }

\affil*[1]{\orgname{Univ. Grenoble Alpes, CNRS, Inria, Grenoble INP, LJK}, \orgaddress{\state{38000 Grenoble}, \country{France}}}

\affil[2]{\orgdiv{Department of Statistics}, \orgname{Brigham Young University}, \orgaddress{\state{Provo, Utah, 84602}, \country{USA}}}


\abstract{Data produced by resting-state functional Magnetic Resonance Imaging are widely used to infer brain functional connectivity networks. Such networks correlate neural signals to connect brain regions, which consist in groups of dependent voxels. Previous work has focused on aggregating data across voxels within predefined regions. However, the presence of within-region correlations has noticeable impacts on inter-regional correlation detection, and thus edge identification. To alleviate them, we propose to leverage techniques from the large-scale correlation screening literature, and derive simple and practical characterizations of the mean number of correlation discoveries that flexibly incorporate intra-regional dependence structures. A connectivity network inference framework is then presented. First, inter-regional correlation distributions are estimated. Then, correlation thresholds that can be tailored to one's application are constructed for each edge. Finally, the proposed framework is implemented on synthetic and real-world datasets. This novel approach for handling arbitrary intra-regional correlation is shown to limit false positives while improving true positive rates.}


\maketitle

\section{Introduction}
	Large-scale network inference is a problem inherent to numerous fields, including gene regulatory networks, spatial data studies, and brain imaging. This work is motivated by an application to resting-state brain functional connectivity networks of single subjects. Such networks connect together correlated brain regions, which consist in groups of dependent voxels. These networks are key to providing insights into the diseased or injured brain \citep{Achard2012Hubsbrainfunctional, RichiardiABV13, malagurski2019topological}. In this paper, the terms \textit{region} and \textit{group} will be used interchangeably, the former being associated with the motivating application, and the second with other data sources of similar structure to which the proposed methods also apply. 
	
	The goal of this work is to infer a binary network where nodes correspond to regions and edges are present only between nodes that are sufficiently highly correlated. The challenge is two-fold: not only does dependence between voxels within a region impact inter-regional correlation estimation, but it also affects inter-regional correlation threshold estimation, and thus edge detection. We propose a correlation screening approach, and tackle the problem of reliable large-scale correlation discovery between two groups of arbitrarily dependent variables.

	In the context of brain functional connectivity, 
	networks are often constructed from functional Magnetic Resonance Imaging (fMRI) data by spatially aggregating blood-oxygen-level-dependent (BOLD) time series within predefined brain regions, e.g., \citep{Fallani2014}. However, this may lead to overestimation of the inter-regional correlation, or inter-correlation for brevity (e.g., \citet{1962halliwellDangeravgcorr}), and hence incorrect edge detection. We propose a novel network inference framework that leverages, for each pair of regions, 
	inter-correlation distributions instead of aggregation. To obtain the associated binary network we then present a thresholding step based on correlation screening. 
	Existing approaches typically assume variables are independent within their region. Yet, as detailed in this work, any violation of this assumption markedly impacts inter-correlation discovery, i.e., when the sample inter-correlation coefficient is greater than a given threshold.
	As will be showcased later, high intra-regional correlation, or intra-correlation for brevity, which corresponds to settings with homogeneous regions, leads to lowered true positive rates (TPR). On the other hand, low intra-correlation, a characteristic of inhomogeneous regions, leads to increased false positive rates (FPR). In \citep{hero2011}, a theoretical framework that accounts for arbitrary dependence is presented. Their approach is nevertheless very difficult to implement in practice and their empirical evaluation only covers the cases of independence or sparse dependence.
	We hence introduce simple and practical expressions to characterize the number of discoveries that flexibly incorporate dependence structures. These can then be employed to find a correlation threshold per pair of regions that improves true discovery rates under dependence, while limiting the number of false discoveries. The main steps of the proposed pipeline are depicted in Figure \ref{fig:pipeline} and are presented in Sections \ref{sec:intercorr}, \ref{sec:corrscreening} and \ref{sec:threshold}. We illustrate our work on synthetic data throughout this paper and 
	demonstrate the effectiveness of our framework on synthetic and real-world brain rat imaging datasets in Section \ref{sec:results}. 
	 \iftrue{
	\begin{figure}[ht]
		\centering
		\includegraphics[width=\linewidth]{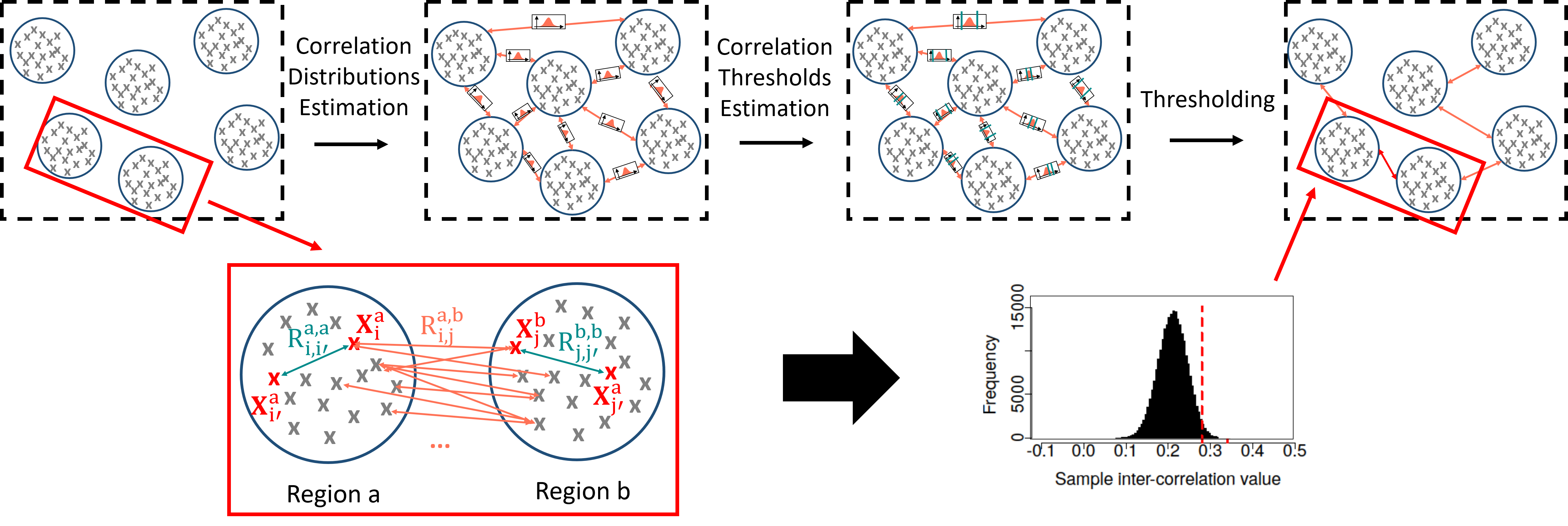}
		\caption{Main steps of our proposed network inference pipeline. Each circle corresponds to a group of variables (represented by crosses). The sample inter-correlation estimation and thresholding steps are detailed for a pair of regions. Some edges were left out to improve readability.} 
		\label{fig:pipeline}
	\end{figure}
 }\fi

	\section{Preliminaries}
	\label{sec:prelim}
	In this section, we define the data model and its parameters that will be used in the rest of this paper. 
	\subsection{Correlation coefficients}
	\label{sec:prelim:corrceof}
	Let $a$ and $b$ be indices of two regions or groups consisting of $p_a$ and $p_b$ random variables, respectively. Denote by $\mathcal{R}_a$ the set of variables in region $a$ and $X_{i}^{a}$ the $i$th random variable in $\mathcal{R}_a$. Assume $n$ independent samples of $X_{i}^{a}$ are available and define the corresponding vector $\textbf{X}_{i}^{a} = [X_{i,1}^{a}, \dots, X_{i,n}^{a}]^T$. $\mathcal{R}_b, X_j^b$ and $\textbf{X}_j^b$ are similarly defined. As an illustration, in the context of brain functional connectivity, $X_{i}^{a}$ corresponds to voxel $i$ of brain region $a$, which is associated with an fMRI BOLD signal time series with $n$ time points. We define \textit{intra-correlation} as the Pearson correlation between each pair of random variables \textit{within} a given region. \textit{Inter-correlation} is the Pearson correlation between pairs of random variables from \textit{two different} regions. 
	
	Let $\rho_{ij}^{a,b}$ denote the true population inter-correlation coefficient between $X_i^a$ and $X_j^b$. For $a \neq b,$ define the corresponding sample inter-correlation coefficient
	\begin{equation}
	    R_{i,j}^{a,b}= \frac{ \sum_{k=1}^{n} (X_{i,k}^{a} - \overline{X_{i}^{a}})(X_{j,k}^{b} - \overline{X_{j}^{b}})}{\sqrt{\sum_{k=1}^{n} (X_{i,k}^{a} - \overline{X_{i}^{a}})^2 \sum_{k=1}^{n} (X_{j,k}^{b} - \overline{X_{j}^{b}})^2} }, 
	\end{equation}
     with $\overline{X_{i}^{a}}$, $\overline{X_{j}^{b}}$ the sample means. Denote the probability density, cumulative distribution, and quantile  functions of $R_{i,j}^{a,b}$ by 	$f_{R_{i,j}^{a,b}}$, $F_{R_{i,j}^{a,b}}$, and $F_{R_{i,j}^{a,b}}^{-1}$, respectively. Population and sample intra-correlation coefficients and their distributions can analogously be defined by choosing $b=a$. Asymptotic closed-form expressions of the density of correlation can be obtained for Gaussian independently identically distributed (i.i.d.) variables $X_i^a$, $X_j^b$ \citep{muirhead2005aspects}. Note however that this work aims to tackle arbitrary dependence between variables, and in this context, to the best of our knowledge, such explicit formulas have not been derived without defining a parametric dependence structure.

In most of this paper, and for ease of calculation, we assume the joint distribution
of pairs of voxels $i, j$ from a fixed region pair $a,b$ are identically distributed. In such cases, the sample inter-correlation coefficients $R_{i,j}^{a,b}$ are identically distributed, and the
$i, j$ subscripts will be dropped, though we emphasize that independence within regions is
not assumed.

	\subsection{Synthetic data examples} 
	\label{exdatadef} 
	We illustrate the different concepts introduced in this paper with data simulated as follows. We consider two regions $a$ and $b$, both containing $p$ intra-correlated variables 
	following a multivariate normal distribution with a predefined Toeplitz covariance structure. $n$ independent samples of each of these $p$ variables are generated. We hence obtain data with a block diagonal covariance matrix of size $2p \times 2p$, where each block corresponds to each region. The off-diagonal blocks correspond to the inter-correlation coefficients, which are set to be constant across all pairs of voxels. The diagonal blocks correspond to the intra-correlation coefficients, which follow a Toeplitz dependence structure.

	\section{Inter-correlation estimation}
	\label{sec:intercorr}
	\subsection{Related work}
	\label{sec:biblio}
	Previous works on the estimation of inter-correlations have mostly focused on aggregating variables within predefined 
 regions \citep{Fallani2014, dadi2019benchmarking}. In the context of brain functional connectivity network inference, some prefer techniques based on independent component analysis (ICA) \citep{calhoun2012exploring}, while most focus on summarizing all voxels within predefined brain regions by their average, e.g., \citep{Achard2012Hubsbrainfunctional, achard.2006.1, di2014autism, malagurski2019topological}. However, such approaches suffer from loss of relevant information and can lead to statistical inconsistency and incorrect correlation estimation \citep{Ostroff19993aggregated}. In particular, the estimate of the average of weakly correlated time series, which corresponds to samples of a single variable in our data model, is poor \citep{1984wigleyAvgtimeseries}. Additionally, it has been observed on small samples that the correlation of averages is different than the average of correlations \citep{1983dunlapAvgcorrCorravg}. This phenomenon can also be easily checked with arbitrary large samples. Furthermore, correlation of averages were empirically observed to overestimate the true correlation \citep{1962halliwellDangeravgcorr, Achard2011fMRIFunctionalConnectivity}. Therefore, when correlating regional averages for binary network inference, one will tend to identify spurious edges. Some previous works attempted to improve false positive rate control utilizing multiple testing approaches \citep{Drton2007mtests}. However, in the context of arbitrary dependence structures, such methods cannot be straightforwardly applied. One alternative to aggregation is to measure the similarity, such as the  Wasserstein distance or covariance \citep{petersen2018WasCov} between intra-correlation densities. However, this approach is not equivalent to 
	that of the Pearson correlation. Indeed,  while the Wasserstein distance may provide a first intuition about how regions are connected, it does not capture as much information about the relationship between the two regions as inter-correlations do. 
	
	\subsection{On the impact of intra-correlation on inter-correlation estimation and detection}
	\label{intraoninter}
 We first illustrate how intra-correlation affects the sample inter-correlation distribution in a simplified scenario, before considering a more general case. It has been known for some time in familial data studies that intra-correlations impact 
	inter-correlation estimation \citep{rosner.1977.1, donner.1991.1}. In the multivariate normal case, and under the assumption of within-group homoscedasticity, the asymptotic variance of the maximum-likelihood estimator of the inter-correlation, 
 denoted as $R_{\mathrm{MLE}}^{a,b}$, was derived by \cite{Elston1975Cor}. This estimator showcases similar behavior to the voxel-to-voxel sample inter-correlation coefficients $R^{a,b}_{i,j}$ and will help provide us with a first intuition about the impact of intra-correlation.
	We need to assume 
 all variables $X_i^a$, $X_j^b$ have the same true inter-correlation $\rho^{a,b}$ and intra-correlation $\rho^{a,a}$ and $\rho^{b,b}$. This amounts to saying sample intra- and inter-correlation coefficients are identically distributed within their corresponding group, or pair of groups, respectively. Under these assumptions, and according to \cite{Elston1975Cor}, the variance of the maximum-likelihood estimator is: 
	\begin{multline}
		\mathrm{Var}(R_{\mathrm{MLE}}^{a,b}) = \frac{1}{n} \left[ (\rho^{a,b})^2 - \frac{1}{p_a} [1 + (p_a - 1)\rho^{a,a}]\right]
		\times \left[ (\rho^{a,b})^2 - \frac{1}{p_b} [1 + (p_b - 1)\rho^{b,b}]\right] \\
		+ \frac{(\rho^{a,b})^2}{2n} \left[\frac{p_a - 1}{p_a}(1-\rho^{a,a})^2 + \frac{p_b - 1}{p_b} (1 - \rho^{b,b})^2 \right]
		\label{varCorr}
	\end{multline}
The expression in \eqref{varCorr} shows that the variance of the sample inter-correlation coefficient explicitly depends on the true intra-correlation coefficients $\rho^{a,a}$ and $\rho^{b,b}$. When the number of samples $n$ is sufficiently large, the inter-correlation variance in the multivariate normal case hence increases when intra-correlation decreases. This observation implies that for a fixed threshold that does not depend on regional dependency structures, more false positive correlations are likely to be discovered.
	
 This intuition is illustrated in the left hand-side of Figure \ref{fig:interhist_varyintra} where the true inter-correlation is zero and no positive correlations are expected to be discovered. Conversely, for the same fixed threshold, when the true inter-correlation is positive (cf. right hand-side of Figure \ref{fig:interhist_varyintra}), increased intra-correlations, which lead to lower inter-correlation variance, may lead to decreased number of true positives. This phenomenon is observed regardless of the number of time points or variables (cf. 
 supplementary materials). 
  \iftrue{

	\begin{figure*}[t]
		\centering
		\includegraphics[width=\linewidth]{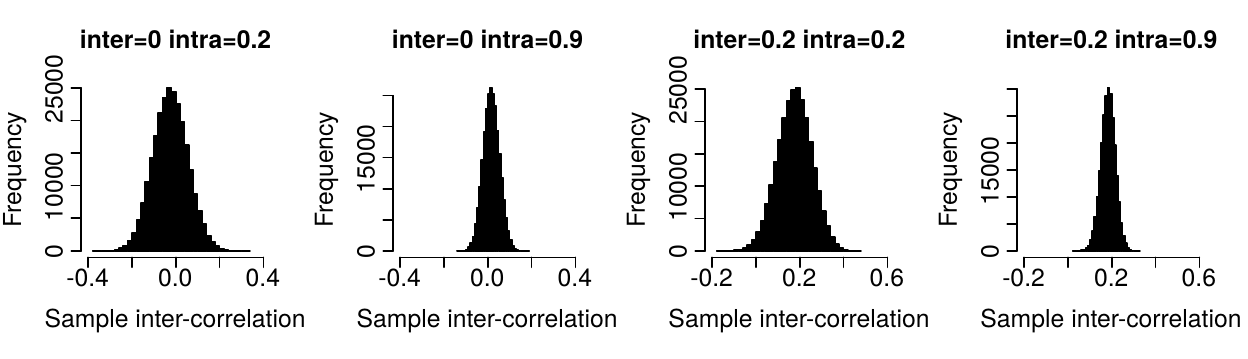} 
		\caption{Effect of intra-correlation on sample inter-correlation distribution, for different population inter-correlation values. The correlation samples were computed between all pairs of variables from two groups. Each group contains $n=150$ samples of $p=500$ intra-correlated random variables following a multivariate normal distribution with Toeplitz intra-correlation (cf. Section \ref{exdatadef}). We can note the higher the intra-correlation, the lower the variance of the inter-correlation distribution. }
		\label{fig:interhist_varyintra}
	\end{figure*}                            
}\fi

In fact, the impact that intra-correlation distributions have on the spatial average of sample inter-correlations can be quantified even without 
 any distributional assumptions. The following result shows that, when the intra-correlation densities of two regions $a$ and $b$ are highly dissimilar, as quantified by a large Wasserstein distance, the average inter-correlation is upper-bounded. In particular, this phenomenon is by no means limited to the Gaussian case. We recall here the definition of the Wasserstein distance between two correlation densities \citep{petersen2018WasCov, statswasserstein2018}:
	$d_W^2(f_{R^{a,a}}, f_{R^{b,b}}) = \int_0^1 [F_{R^{a,a}}^{-1}(c) - F_{R^{b,b}}^{-1}(c)]^2 dc$. The full proof is available in Appendix \ref{append:wassinter}.  
	\begin{proposition}
		\label{prop:wassdist}
		For any regions $a$, $b$ with $p_a, p_b$ voxels, respectively, and $\mathcal{R}_a, \mathcal{R}_b$ the corresponding voxel sets, if there exists $A \in \mathbb{R}^{+}$ such that $d_W^2(f_{R^{a,a}}, f_{R^{b,b}}) \geq \min\limits_{c\in [0,1]} \left( F_{R^{a,a}}^{-1}(c) - F_{R^{b,b}}^{-1}(c) \right)^2 \geq A$, then,
		$\overline{R^{a,b}} = \frac{1}{p_a p_b} \sum\limits_{i \in \mathcal{R}_a} \sum\limits_{j \in \mathcal{R}_b}  R^{a,b}_{i,j} \leq 1 - \frac{\sqrt{A}}{2}$.
	\end{proposition}

	\subsection{Proposed approach: inter-correlation distribution estimation}
	As previously discussed, aggregating variables within regions to estimate inter-correlation leads to loss of information and incorrect edge detection during the binary network inference step. 
	We have also brought to light the importance of taking intra-correlation into account when manipulating inter-correlations. In addition, all the previously cited approaches that aim 
	to infer a binary network 
	where nodes are groups of variables only provide a single correlation threshold to be applied to all pairs of regions. In this paper, we propose to derive a correlation threshold specific to each pair of regions to better harness the particularities of the regional dependence structures. Instead of averaging variables 
	within regions, we hence propose to estimate the distribution of correlations measured between all pairs of variables from two different regions. We then obtain an inter-correlation distribution per pair of region, which then needs to be thresholded. To that end, we propose to leverage correlation screening.
In that paradigm, an edge is said to be detected in the associated binary graph if 
a sufficient number of sample 
inter-correlation coefficients of the corresponding pair of regions are large enough. In the following sections, we derive simplified expressions of the number of discoveries to propose a reliable method to threshold these inter-correlation distributions.
	
	\section{Characterization of the number of discoveries under dependence} 
	\label{sec:corrscreening}
	
	Correlation screening \citep{hero2011}, or independence screening \citep{2008Indscreening}, is often used in variable or feature selection problems. In such approaches, the goal is to discover sufficiently highly correlated variables. A practical method consists in defining a correlation threshold, above which correlation coefficients, and their associated variables, are said to be \textit{detected} or \textit{discovered}. Nonetheless, in high dimension, such approaches may suffer from a high number of false discoveries. In \citep{hero2011}, the authors 
	aim to mitigate this issue in the following way. They first propose the following maximum-based definition of the number of discoveries, pertaining to inter-correlation coefficients, with $\phi_{ij}^{ab}(\rho) = \mathbbm{1}(| R^{a,b}_{i,j}| > \rho)$ for all voxels $i,j$ in regions $a,b$ and correlation threshold $\rho \in [0,1]$:
	\begin{equation}
		N^{ab}(\rho) = \sum\limits_{i=1}^{p_a} \max\limits_{j =1, \dots ,p_b } \phi_{ij}^{ab}(\rho).
		\label{Nab}
	\end{equation}
	The authors provide as well an approximation of the expected number of discoveries $E[N^{ab}]$ that depends on the number of variables $p$, the number of samples $n$ and a function of the joint distribution of a transformation of the variables. Then they employ the derived formula to compute critical threshold values based on a phase transition approach. Furthermore, the expected number of discoveries 
	is used to control the number of false discoveries. It is then all the more essential to have %
	an expression of the number of discoveries that is both interpretable and can easily be theoretically and empirically utilized. 
	However, the expression for 
	$E[N^{ab}]$ derived in \citep{hero2011}, which depends on joint distributions, is difficult to compute, especially for single subject analysis where we have access to only a single sample of each signal. We hence provide simplified explicit expressions of the mean number of discoveries that still harness information contained in the intra-correlation distributions. 
	
	\subsection{Maximum-based expression: $N^{ab}$}
	Empirically, intra-correlation has an impact on $N^{ab}$ and its average (cf. Figure \ref{fig:Nabiid}, left). Indeed, for a given inter-correlation threshold, the smaller the intra-correlations, the larger the number of discoveries. This is in accordance with the observations from Figure \ref{fig:interhist_varyintra}. Additionally, we can remark that in this example the true inter-correlation is zero. Thus any discovery is a false positive.
	As the intra-correlation increases, a lower correlation threshold is then sufficient to maintain similar levels of false discoveries. 
	
  \iftrue{

	\begin{figure}[htb]
		\centering
		\begin{subfigure}{0.48\textwidth}
			\centering
			\includegraphics[trim=0cm 1cm 0cm 2cm, clip=true, width=\linewidth]{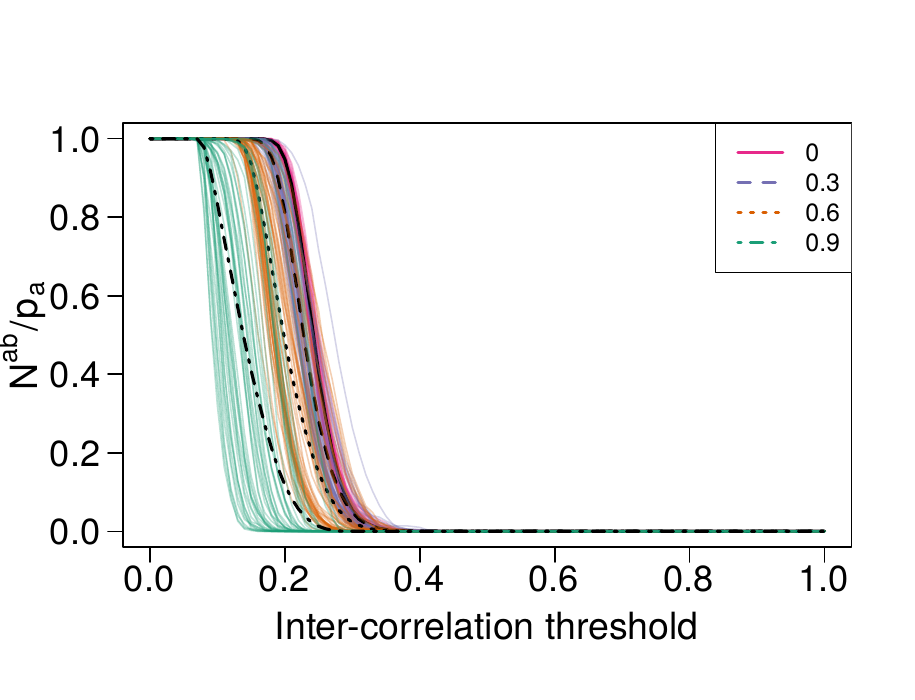}
		\end{subfigure}
		~
		\begin{subfigure}{0.48\textwidth}
			\centering
			\includegraphics[trim=0cm 1cm 0cm 2cm, clip=true,width=\linewidth]{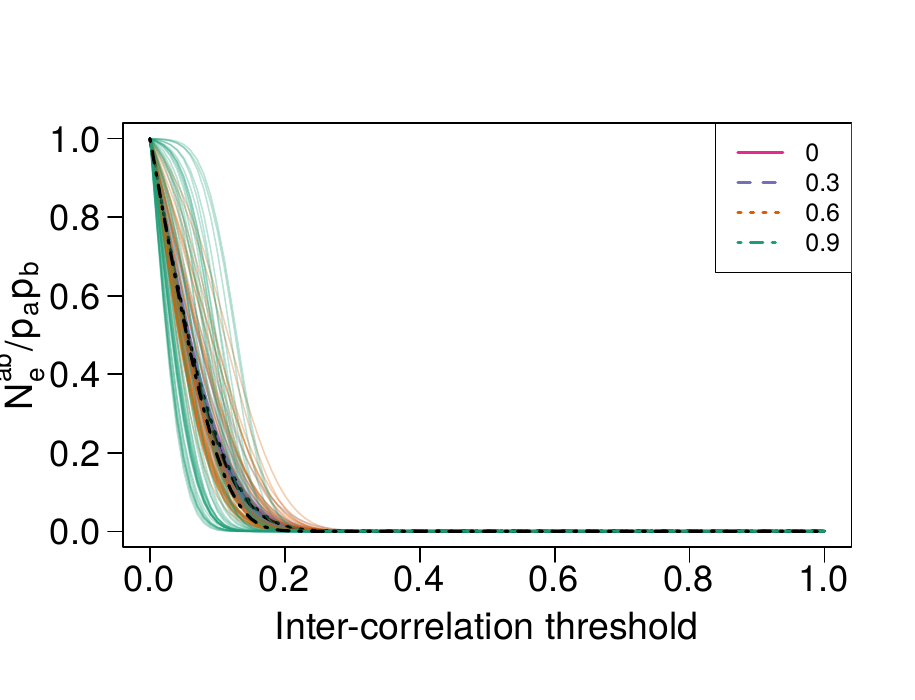}
		\end{subfigure}
		\caption{Normalized number of discoveries, $N^{ab}/p_a$ (Left) and $N_e^{ab}/p_a p_b$ (Right) as a function of the inter-correlation threshold for data simulated as described in Section \ref{exdatadef}, with $p_a = p_b = 500, n=150$, true inter-correlation $\rho^{ab}=0$ and Toeplitz intra-correlation with varying minimal intra-correlations. For each of the four intra-correlation values, $50$ datasets were simulated and used to compute the number of discoveries (the colored curves), and their average (the black dotted curves). $N^{ab}/p_a$ decreases as the intra-correlation increases, while $N_e^{ab}/p_a p_b$ does not seem to be much impacted on average, even though its variability seems to increase with the intra-correlation value.}
		\label{fig:Nabiid}
	\end{figure}
	}\fi

	From (\ref{Nab}), we can also conclude that for any $\rho \in [0,1]$,
	\begin{equation}
	     E[N^{ab}(\rho)] = \sum_{i=1}^{p_a} \left( 1 -  F_{| R_{i,1}^{a,b}|, \dots , | R_{i,p_b}^{a,b}|}(\rho,...,\rho) \right), 
	\end{equation}    
	 with $F_{| R_{i,1}^{a,b}|, \dots , | R_{i,p_b}^{a,b}|}$ the joint distribution of the absolute values of the corresponding correlation coefficients, which will inherently take into account dependence structures between the inter-correlation coefficients. However, joint distributions are complicated to estimate and manipulate. Let $\Tilde{\nu}^{ab} = E[N^{ab}]/p_a$.	We then propose an approximate expression of $\Tilde{\nu}^{ab}$, denoted by $\nu^{ab}$, that depends on the distribution of inter-correlations, and that is exact under some particular assumptions (cf. Appendix \ref{append:Nab}): 
	\begin{equation}
		\nu^{ab}(\rho) =  1 - F_{|R^{a,b}|}(\rho)^{p_b} ,  \quad \rho \in [0,1].
	\end{equation} 
	
	 We can also derive the following inequality.  
	\begin{proposition} Consider two regions $a$ and $b$ and a correlation threshold $\rho \in [0,1]$. If all variables $X_i^a$ and $X_j^b$ in both regions follow a normal distribution and their sample inter-correlation coefficients are identically distributed, then for sufficiently large $n$,
		\label{prop:ineqNab}
		\begin{equation}
			\nu^{ab}(\rho) \geq \Tilde{\nu}^{ab}(\rho)
			\label{ineqNab}
		\end{equation}
		\label{prop:Nabineq}
	\end{proposition}
	\begin{proof}
		As defined in \citep{1966LehmanPQD}, the random variables $T_1, T_2, \dots, T_p$ are Positively Quadrant Dependent (PQD) if for any positive number $t_1, t_2, \dots, t_p$, 
		\begin{equation}
			\label{PQDdef}
			P\left(\bigcap\limits_{k=1}^{p} T_k \leq t_k \right) \geq \prod\limits_{k=1}^{p} P \left( T_k \leq t_k  \right).        
		\end{equation}
		
		Under the assumption $X_i^{a}, X_j^{b}$ are normal for all $i,j$, the distribution of their sample correlation coefficients $R_{i,j}^{a,b}$ is asymptotically normal, e.g., \citep{Ruben1966, 1953Hotelling}. Hence, according to Theorem 1 in \citep{1967SidakPQD}, when $n$ is large enough, $|R_{i,j}^{a,b}|$ are PQD. We can then note that equation (\ref{PQDdef}) is equivalent to $ E\left[\prod_{k=1}^{p}\mathbbm{1}(| T_{k}| \leq t_k) \right] \geq  \prod_{k=1}^{p} E[\mathbbm{1}(| T_{k}| \leq t_k)] $. Under the assumption the sample correlation coefficients are identically distributed, and setting $t_k=\rho$ and $T_k=R
  ^{a,b}$ for all variables, we can thus write:
		\begin{align*}
		p_a \cdot \nu^{ab}(\rho) = & \sum\limits_{i=1}^{p_a} \left( 1 - \prod_{j=1}^{p_b} E[\mathbbm{1}(| R
  ^{a,b}| \leq \rho)] \right) \\ \geq & \sum\limits_{i=1}^{p_a} \left( 1 - E\left[\prod_{j=1}^{p_b} \mathbbm{1}(| R
  ^{a,b}| \leq \rho)\right] \right) =  E[N^{ab}(\rho)] = p_a \cdot \Tilde{\nu}^{ab}(\rho).
		\end{align*}
   This concludes the proof.
	\end{proof}
	This result ensures $\nu^{ab}$ will provide thresholds that are at least as conservative as that of $\Tilde{\nu}^{ab}$. Moreover, the assumptions needed in Proposition \ref{prop:ineqNab} are often reasonable in practice---and notably in functional brain connectivity applications where the signals associated with each voxel can appropriately be transformed \citep{whitcher.2000.2}. In addition, $\nu^{ab}$ can be estimated by $\widehat{\nu}^{ab}(\rho) = 1 - \widehat{F}_{|R^{a,b}|}(\rho)^{p_b}$, where $\widehat{F}_{| R^{a,b}|}(\rho) = \frac{1}{p_b p_a} \sum_{i=1}^{p_a} \sum_{j=1}^{p_b} \mathbbm{1}(| R^{a,b}_{i,j}| \leq \rho)$ is the empirical cumulative distribution function (ecdf) of $| R^{a,b}|$.
	The result stated above can then be empirically observed in Figure \ref{fig:NabNeiid}. The curve of $\widehat{\nu}^{ab}$ as a function of thresholds is also particularly close to that of the empirical values of $\Tilde{\nu}^{ab}$ as long as both the inter-correlation and the intra-correlation of region $b$ are not too high. $\widehat{\nu}^{ab}$ provides hence an approximation for the normalized expected number of discoveries that is easier to compute, while still accounting for the inter-correlation distribution. Moreover, in practice, the ecdf of the inter-correlation coefficients can be shown to depend on the intra-correlation structure \citep{2014AzrielEcdf} and hence allows us to account for it. 
	 \iftrue{

	\begin{figure*}[!ht]
		\centering
		\includegraphics[width=\textwidth]{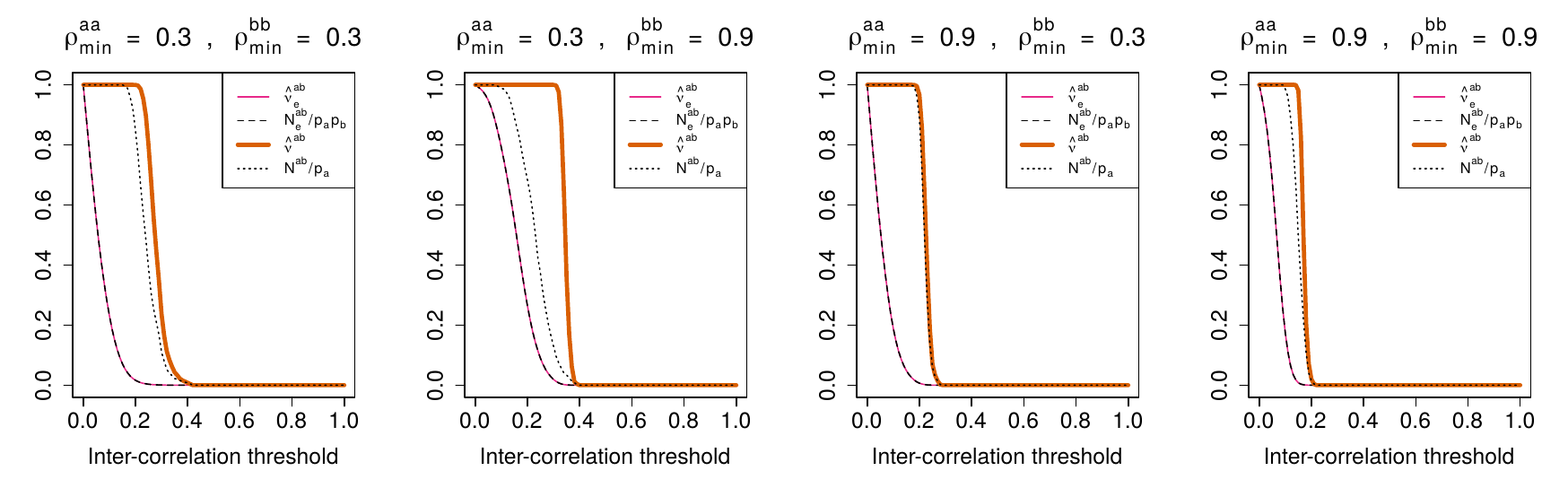}
		\caption{Normalized number of discoveries $\widehat{\nu}^{ab}$, $\widehat{\nu}_e^{ab}$, $N^{ab}/p_a$ and $N_e^{ab}/p_a p_b$ as a function of the correlation threshold for different minimal intra-correlation values $\rho_{min}^{aa}, \rho_{min}^{bb}$ and a zero inter-correlation for data simulated as described in Section \ref{exdatadef}, with $p_a = p_b = 500, n=150$. We can note $\widehat{\nu}_e^{ab} \leq \widehat{\nu}^{ab}$. We can also observe the asymmetric behavior of $N^{ab}/p_a$ and how it is close to $\widehat{\nu}^{ab}$ when $\rho^{aa}_{min} \gg \rho^{bb}_{min}$, unlike when $\rho^{bb}_{min} \gg \rho^{aa}_{min}$.}
		\label{fig:NabNeiid}
	\end{figure*}
	}\fi
	\subsection{Sum-based expression: $N_e^{ab}$}
	We now present another, and more intuitive, way to characterize the number of discoveries $N_e^{ab}$ \citep{hero2011}. It represents the total number of discoveries and will enable us to propose less conservative thresholds:
	\begin{equation}
		N_e^{ab}(\rho) = \sum_{i=1}^{p_a} \sum_{j=1}^{p_b} \phi_{ij}^{ab}(\rho), \quad \rho \in [0,1].
		\label{NeDef}
	\end{equation}
	However, it is not straightforward to derive a critical correlation threshold from this expression. We propose the simplified expression below: 
	\begin{equation}
	   \widehat{\nu}_e^{ab}(\rho) =  1 - \widehat{F}_{|R^{a,b}|}(\rho), \quad \rho \in [0,1].	
	\end{equation} 

We can also remark in Figure \ref{fig:NabNeiid} that $\widehat{\nu}_e^{ab}$ and $N_e^{ab}/p_a p_b$
 look indistinguishable.

	\subsection{Link between $N_e^{ab}$ and $N^{ab}$}
	We have presented so far two ways to characterize the number of discoveries. We will now discuss how they relate to one another. We can remark the following inequality. 
	\begin{proposition}
		For all $\rho \in [0,1]$, 
		\begin{equation}
		    \widehat{\nu}_e^{ab}(\rho) 
      \leq \widehat{\nu}^{ab}(\rho). 
		\end{equation}
		\label{prop:NeNabineq}
	\end{proposition}
	The proof is straightforward and can be found in Appendix \ref{append:NeNabineq}. This result can notably be observed in Figure \ref{fig:NabNeiid}. Thus $\widehat{\nu}^{ab}$ is more conservative than $\widehat{\nu}^{ab}_e$, which in some circumstances may be desirable. Nonetheless, when the inter-correlation is zero we expect no discoveries. In this case, a critical correlation threshold can hence be defined as the minimum correlation such that the number of discoveries is zero. In such cases, using $\widehat{\nu}_e^{ab}$ then seems to be preferable, since it provides a lower correlation threshold for a similar number of false discoveries.

	\section{Correlation threshold definition} \label{sec:threshold}
	Now we have better characterized the number of discoveries, we can use it to construct correlation thresholds tailored to one's data and that ensure, to a certain extent, a restricted number of false discoveries and improved number of true discoveries. We present in this section two possible correlation threshold definition approaches. The idea behind correlation threshold definition is to ensure that
	it is very unlikely for any discovery to correspond to a correlation value that could have happened at random. 
 This amounts to a setting where the true inter-correlation is zero, which may not be true in practice. Surrogate data defined such that the population inter-correlation is zero can hence be utilized to estimate the correlation thresholds. We denote $\widehat{F}^{-1}_{0, |R^{a,b}|}$ the corresponding  quantile function.
 
	\subsection{FWER-based threshold} Correlation thresholds with family-wise error rate (FWER) theoretical control can be derived for specific dependence structures. These approaches control the probability of making at least one false discovery. Analogously to Proposition 2 in \citep{hero2011}, it can be shown, under a weak dependence condition, that $N_e^{ab}$ converges to a Poisson random variable when $p_a, p_b \to \infty$ and  $P(N_e^{ab} > 0) \to 1 - \exp\left(- E[N_e^{ab}]\right)$. Our proposed expression $\widehat{\nu}_e^{ab}$ can then be used to compute correlation thresholds $\rho_\alpha^{ab}$ that guarantee a FWER at level $\alpha$. Nevertheless, the weak dependence assumption upon which this approach hinges is often not reasonable in practice.
	
	\subsection{Quantile-based threshold} 

The correlation threshold can also be defined such that the False Positive Rate (FPR) is guaranteed to be less than a given level $\alpha$. The FPR is the ratio between the number of false positives (FP) and the total number of ground truth negatives, that is $p_a \cdot p_b$ in the $\rho^{ab}=0$ case. Controlling the FPR at level $\alpha$ is thus equivalent to ensuring the number of discoveries is less than $FP=\alpha \cdot p_a \cdot p_b$. Since in our setting ($\rho^{ab}=0$) any discovery is a false positive, we can set $\widehat{\nu}^{ab}_e \cdot p_a \cdot p_b = \alpha \cdot p_a \cdot p_b$, which leads to the threshold $\rho^{ab}_{q, \alpha} = \widehat{F}^{-1}_{0,|R^{a,b}|}(1 - \alpha) $. We can remark that when $\alpha = 0$, the chosen threshold is larger than any of the observed absolute correlations, ensuring there will be no discoveries. Additionally, this threshold will depend on the intra-correlation, as does the ecdf \citep{2014AzrielEcdf}. We can also remark this threshold guarantees a FWER at level $\alpha = 0$ under the previous weak dependence assumption. $\widehat{\nu}^{ab}$ can also be used to similarly derive a critical threshold, although stringent conditions on the region sizes would then need to be verified when $\alpha \neq 0$.

 \begin{figure}[h]
     \centering
     \includegraphics[width=.8\linewidth]{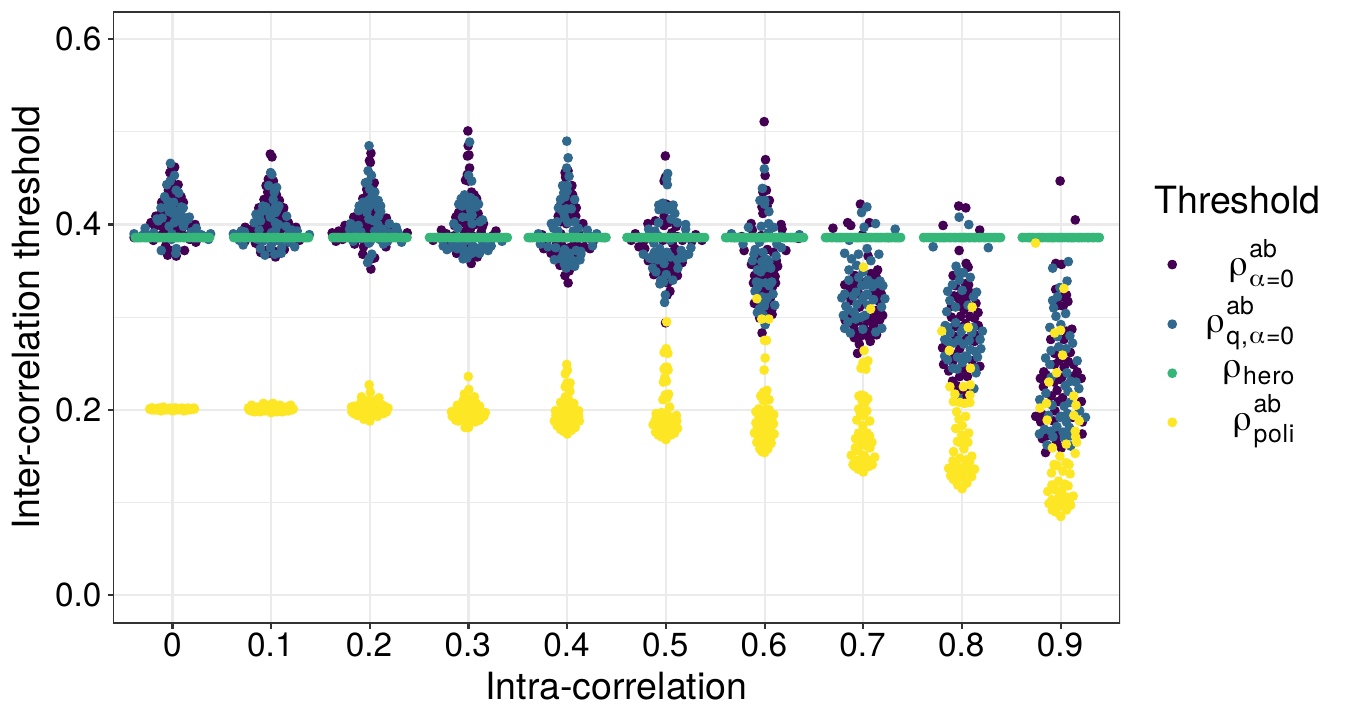}
     \caption{Comparison of different critical thresholds for 50 replicates of data simulated as described in Section \ref{exdatadef} with $p=150$, $n=100$, $\rho^{a,b}=0$ and varying constant intra-correlation values. The FWER- and quantile-based thresholds $\rho_\alpha^{ab}$ were computed for $\alpha=0$.}
     \label{fig:boxplots}
 \end{figure}
 
	\subsection{Numerical results} We compare in Figure \ref{fig:boxplots} the two correlation thresholds defined above with two other approaches: the critical thresholds $\rho_{hero}$ 
	obtained in \citep{hero2011}, and a simple method where the threshold is set to $\rho^{ab}_{poli} = \widehat{\mu}^{ab}_0 + \widehat{\sigma}^{ab}_0$ where $\widehat{\mu}^{ab}_0$ and $\widehat{\sigma}^{ab}_0$ are the sample mean and standard deviation of the sample inter-correlation of the surrogate data \citep{2015Poli}. We observe that, when intra-correlation is lower than $0.5$, both our proposed approaches provide similar thresholds to \cite{hero2011}. Nonetheless, our FWER- and quantile-based methods provide less conservative thresholds when the intra-correlation is high. The lower thresholds imply an increase in the true positive rate. We can note as well that $\rho^{ab}_{poli}$ is lower than all three other thresholds for all intra-correlation values, which could lead to a large number of false discoveries, as will be shown in the next section. It is also decreasing when intra-correlation increases in accordance with the observations about the effect of intra-correlation on the distribution of sample inter-correlation in Section \ref{sec:intercorr}.

	\section{Network inference results} 
	\label{sec:results}
	In this section we provide an illustration of our network inference approach on synthetic and real-world data and compare it to several methods.

 \subsection{Comparison to other methods}
	 As mentioned in Section \ref{sec:biblio}, most single-subject fMRI studies that use anatomical parcellations estimate the inter-correlation by computing the correlation coefficient between spatial regional averages of the signals. We will refer to the correlation of averages approach by CA and our proposed correlation screening method by CS. Various methods can then be employed to define the correlation thresholds. They usually belong to either of these two paradigms: (i) \textit{relative thresholding}, that is, estimation of the binary network by extracting a fixed proportion of edges, e.g., \citep{VANDENHEUVEL2017437}, and (ii) \textit{absolute thresholding} where one chooses, more or less arbitrarily, a fixed threshold that will be applied to all edges---as opposed to our proposed edge-specific thresholds. These two approaches are implemented in several popular packages, such as the Brain Connectivity toolbox \citep{RUBINOV20101059} and CONN \citep{conn}, with little guidance on the choice of threshold. Relative thresholds will always detect the same proportion of edges regardless of the true inter-correlation value, and as such can be disregarded in this work. Both the thresholds proposed by \cite{hero2011} and \cite{2015Poli} are absolute thresholds. The latter was used in \citep{2021Boschi_thresh} to threshold CA-based functional connectivity. In \citep{guillaume2020functional}, the authors apply to all edges of a CA-based functional connectivity network a fixed threshold $\rho_{becq}$ determined according to a multiple testing approach (see Appendix G of their work for more details). In practice, the values of $\rho_{becq}$ are close to that of $\rho_{hero}$. The thresholds $\rho^{ab}_{poli}, \rho^{ab}_{\alpha}$ and $\rho^{ab}_{q, \alpha}$ are estimated using surrogate data where the true inter-correlation is zero and the intra-correlation is constant and equal to the average sample intra-correlation. 

		
	\subsection{Synthetic data results}
We generated synthetic datasets with ten inter-connected regions. For each dataset, $10$ regions are simultaneously simulated, each region containing $p=150$ intra-correlated variables following
a multivariate normal distribution. $100$ independent samples of each of these variables are obtained. For each region, a Toeplitz intra-correlation is used with the same minimal intra-correlation value across all ten regions. There are $41$ true positive (constant true inter-correlation $\rho^{ab}=0.2$) and $4$ true negative edges (constant inter-correlation $\rho^{ab}=0$) in the ground-truth network. The different simulation parameters are chosen to ensure the population covariance matrix of the ten regions is positive semidefinite. To identify the edges of the binary network, the pairwise thresholds are applied to the corresponding distributions of the absolute value of the sample inter-correlation. In particular, pairs of regions where the inter-correlation is larger than the threshold with a probability at most $0.05$ are not identified as edges.

 \begin{sidewaystable}[ph!]
    \centering 
	\caption{Comparison of the mean (standard deviation) TPR and FPR of several network inference methods on synthetic data across $100$ repetitions. The ground-truth networks consist of $10$ regions, with $4$ real negative edges (true inter-correlation $\rho^{ab}=0$), $41$ real positive edges (true inter-correlation $\rho^{ab}=0.2$), five minimal intra-correlation values $\rho_{min}^{aa}$, $n=100$, and $p=150$. }	
 \label{table:TPFP}
\begin{tabular*}{\textheight}{@{\extracolsep\fill}lcccccccccc}
			\toprule
			& \multicolumn{2}{c}{$\rho_{min}^{aa} = 0.5$}
            & \multicolumn{2}{c}{$\rho_{min}^{aa} = 0.6$}
            & \multicolumn{2}{c}{$\rho_{min}^{aa} = 0.7$}
            & \multicolumn{2}{c}{$\rho_{min}^{aa} = 0.8$}
			& \multicolumn{2}{c}{$\rho_{min}^{aa} = 0.9$} \\
			\cmidrule(r){2-3} \cmidrule(l){4-5} \cmidrule(l){6-7} \cmidrule(l){8-9} \cmidrule(l){10-11}  
			Method & FPR & TPR & FPR & TPR & FPR & TPR & FPR & TPR & FPR & TPR  \\
			\midrule
			CA + $\rho^{ab}_{poli}$ & 0.51 (0.23) & 0.97 (0.02) & 0.46 (0.22) & 0.90 (0.07) & 0.45 (0.25) & 0.69 (0.12) & 0.43 (0.26) & 0.33 (0.14) & 0.30 (0.24) & 0.04 (0.03) \\
			CA + $\rho_{becq}$ & 0.13 (0.13) & 0.47 (0.20) & 0.04 (0.09) & 0.22 (0.13) & 0.03 (0.08) & 0.11 (0.09) & 0.01 (0.04) & 0.05 (0.06) & 0.01 (0.04) & 0.04 (0.04) \\
   			CS + $\rho^{ab}_{poli}$ & 1.00 (0.00) & 1.00 (0.00) & 1.00 (0.00) & 1.00 (0.00) & 1.00 (0.00) & 1.00 (0.00) & 1.00 (0.00) & 1.00 (0.00) & 1.00 (0.00) & 1.00 (0.00) \\
			CS + $\rho_{hero}$ & 0.03 (0.08) & 0.13 (0.12) & 0.03 (0.08) & 0.13 (0.11) & 0.04 (0.09) & 0.12 (0.09) & 0.02 (0.07) & 0.11 (0.09) & 0.01 (0.06) & 0.09 (0.07) \\
			CS + $\rho^{ab}_{\alpha=0}$ & 0.06 (0.10) & 0.23 (0.17) & 0.08 (0.12) & 0.32 (0.17) & 0.13 (0.13) & 0.42 (0.17) & 0.15 (0.16) & 0.52 (0.16) & 0.23 (0.17) & 0.68 (0.13) \\
			CS + $\rho^{ab}_{q, \alpha=0}$ & 0.06 (0.11) & 0.23 (0.18) & 0.07 (0.11) & 0.32 (0.17) & 0.14 (0.13) & 0.42 (0.16) & 0.15 (0.16) & 0.53 (0.16) & 0.25 (0.18) & 0.67 (0.13) \\
			\botrule
		\end{tabular*} 
\end{sidewaystable}

Table \ref{table:TPFP} displays the false positive and true positive rates (FPR and TPR, respectively) of the different methods, for varying minimal intra-correlation values. The FPR and TPR are defined as follows: FPR = FP/(FP+TN) and TPR = TP/(TP+FN). FN stands for false negatives (i.e., an edge is undetected when it actually exists). FPR is expected to be close to $0$ and TPR to $1$.	Results in Table \ref{table:TPFP} showcase that, as expected from the previous section, using $\rho_{poli}^{ab}$ leads to high FPRs, while $\rho_{hero}$ and $\rho_{becq}$ lead to decreasing TPRs as intra-correlation increases. Additionally, while the FPR is slightly increased, correlation screening methods with FWER- or quantile-based thresholds markedly improve the TPR when the  intra-correlation is high, and should be preferred in that case. Indeed, when intra-correlation is $0.9$ all other methods (except CS + $\rho_{poli}^{ab}$) have a TPR close to zero, while the FWER- and quantile-based thresholds have a TPR close to $0.7$. The CS+$\rho_{poli}^{ab}$ method displays a FPR of $1$ for all intra-correlation values and should thus to be avoided. Since the FWER- and quantile-based thresholds are empirically equivalent, from now on we will be using $\rho_{q, \alpha = 0}^{ab}$, which, unlike $\rho_{\alpha=0}^{ab}$, has a theoretical control over false positives that is valid for any dependence structure. 
 
	
\subsection{Real-world data results}
    We applied our framework on functional Magnetic Resonance Imaging (fMRI)
	data acquired on both dead and live rats, anesthetized using Isoflurane \citep{becq_10.1088/1741-2552/ab9fec,guillaume2020functional}.
	The datasets are freely available at \url{https://dx.doi.org/10.5281/zenodo.7254133}. The scanning duration was $30$ minutes with a time
	repetition of $0.5$ second so that $3600$ time points were acquired. After preprocessing as explained in \citep{guillaume2020functional}, based on an anatomical atlas, $51$ groups of time series, corresponding to the rat brain regions, were extracted for each rat. Due to insufficient signal, four regions were excluded. Each time series captures the functioning of a given voxel. The dead rats provide experimental data where the ground-truth network is empty. Indeed, no legitimate functional activity should be detected, whereas for the live rat under anesthetic, we expect non-empty graphs as brain activity keeps on during anesthesia. We can note that no ground-truth is available for the live rat networks. As expected, the networks of the dead rats estimated using our proposed correlation screening method are empty, i.e. it does not detect any false positive edges, with the exception of one edge in one rat. However, the CA + $\rho_{poli}^{ab}$ approach \citep{2021Boschi_thresh,2015Poli} detects over $300$ false positive edges, and \cite{guillaume2020functional} (later denoted B2020) detect between one and four false positive edges (cf. Table \ref{table:rats}). While our approach is more conservative than the other two, important edges are still detected in the live rats, mainly in motor regions (M1 and M2) and somatosensory regions (S1 and S2), as shown for instance in Figure \ref{fig:ratnetwork}.

	 \iftrue{

		\begin{table*}[!h]
		\caption{Comparison of the number of edges in the networks obtained via our proposed network inference approach and two methods from the literature (B2020 and CA + $\rho_{poli}^{ab}$) for dead (Top) and live (Bottom) rat brain fMRI data. In the dead rat brain networks, any detected edge is a false positive.}
		\label{table:rats}
		\centering
		\begin{subtable}[ht]{\textwidth}
			\centering
			\begin{tabular}{lccc}
				\toprule
				& \multicolumn{3}{c}{\textbf{NUMBER OF EDGES}} \\
				\cmidrule(r){2-4} 
				\textbf{DEAD RATS ID} 
				& \textbf{CS + $\mathbf{\rho_{q, \alpha=0}^{ab}}$} &  B2020 & CA + $\rho_{poli}^{ab}$ \\
				\midrule
				$20160524\_153000$ & $\mathbf{1}$ & $4$ & $316$ \\
				$20160609\_161917$ & $\mathbf{0}$ & $4$ & $317$ \\
				$20160610\_121044$ & $\mathbf{0}$ & $1$ & $325$ \\
				\botrule
		\vspace{.1cm}
			\end{tabular}
		\end{subtable}
		\begin{subtable}[hb]{\textwidth}
		    \centering
			\begin{tabular}{lccc}
				\toprule
				& \multicolumn{3}{c}{\textbf{NUMBER OF EDGES}} \\
				\cmidrule(r){2-4} 
				\textbf{LIVE RATS ID} 
				&  $\mathbf{\text{\textbf{CS}} + \rho_{q, \alpha=0}^{ab}}$ & B2020 & CA + $\rho_{poli}^{ab}$ \\
				\midrule
				$20160615\_103000$ & $\mathbf{25}$ & $647$ & $820$ \\
				$20160614\_095825$ & $\mathbf{411}$ & $847$ & $910$ \\
				$20160615\_121820$ & $\mathbf{116}$ & $477$ & $692$ \\
				$20160421\_133725$ & $\mathbf{83}$ & $591$ & $910$ \\
				\botrule
			\end{tabular}
		\end{subtable}
	\end{table*}
	
 \begin{figure}[h]
        \centering
        \begin{subfigure}[b]{.3\textwidth}
            \centering
            \includegraphics[trim={1cm 4cm 1cm 3.5cm},clip, width = .9\textwidth]{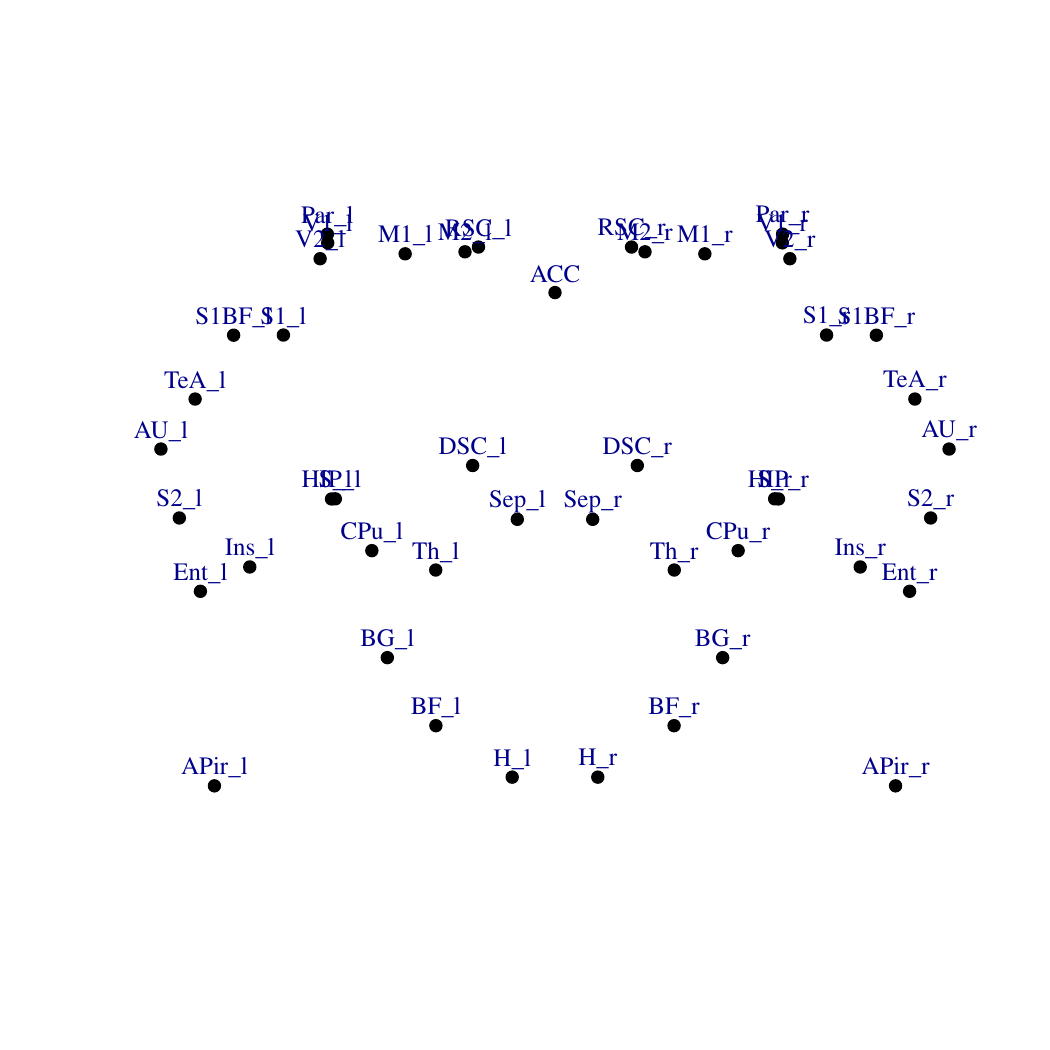}
            \caption{Dead rat $20160609\_161917$}
        \end{subfigure}
        ~
        \begin{subfigure}[b]{.3\textwidth}
            \centering
            \includegraphics[trim={1cm 4cm 1cm 3.5cm},clip, width = .9\textwidth]{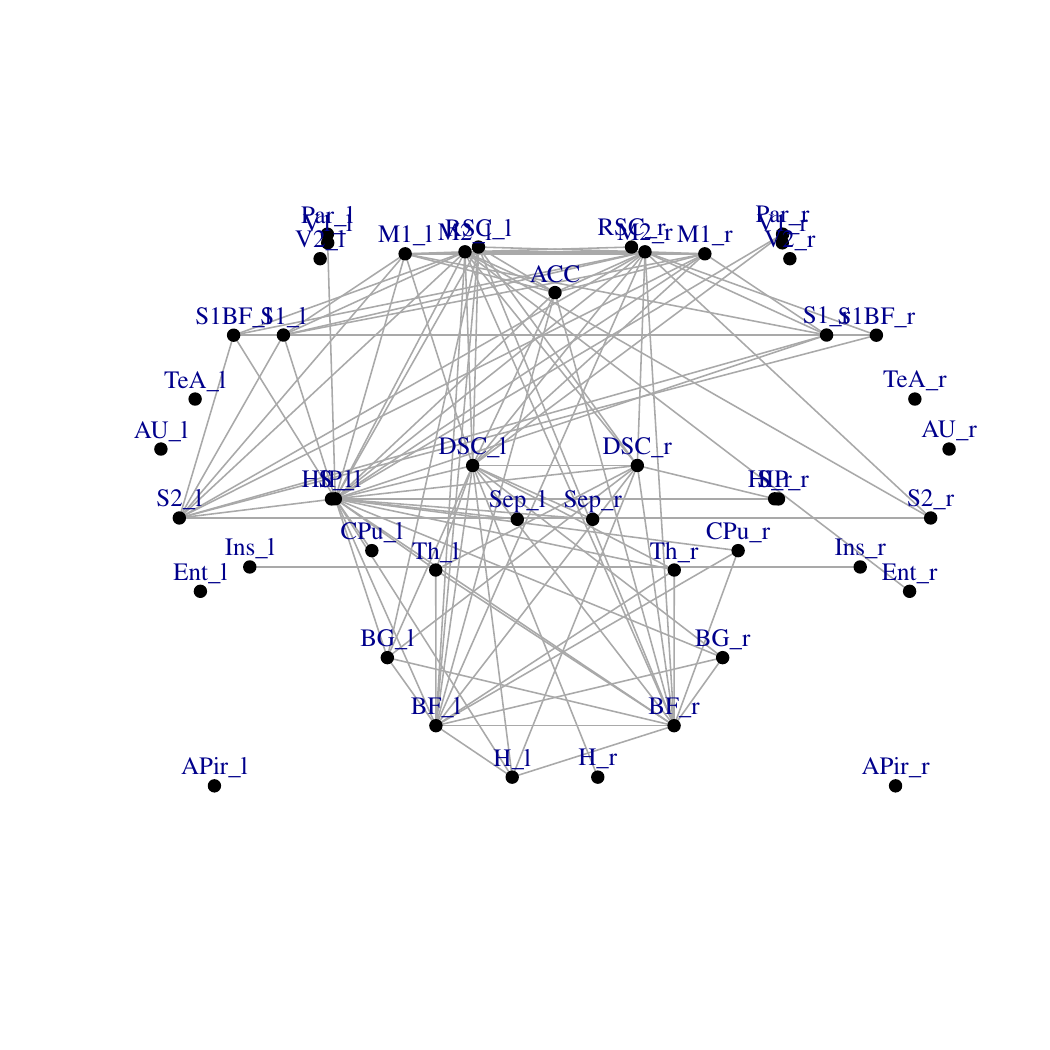}
            \caption{Live rat $20160615\_121820$ }
        \end{subfigure}
        \begin{subfigure}[b]{.3\textwidth}
            \centering
            \includegraphics[trim={1cm 4cm 1cm 3.5cm},clip, width = .9\textwidth]{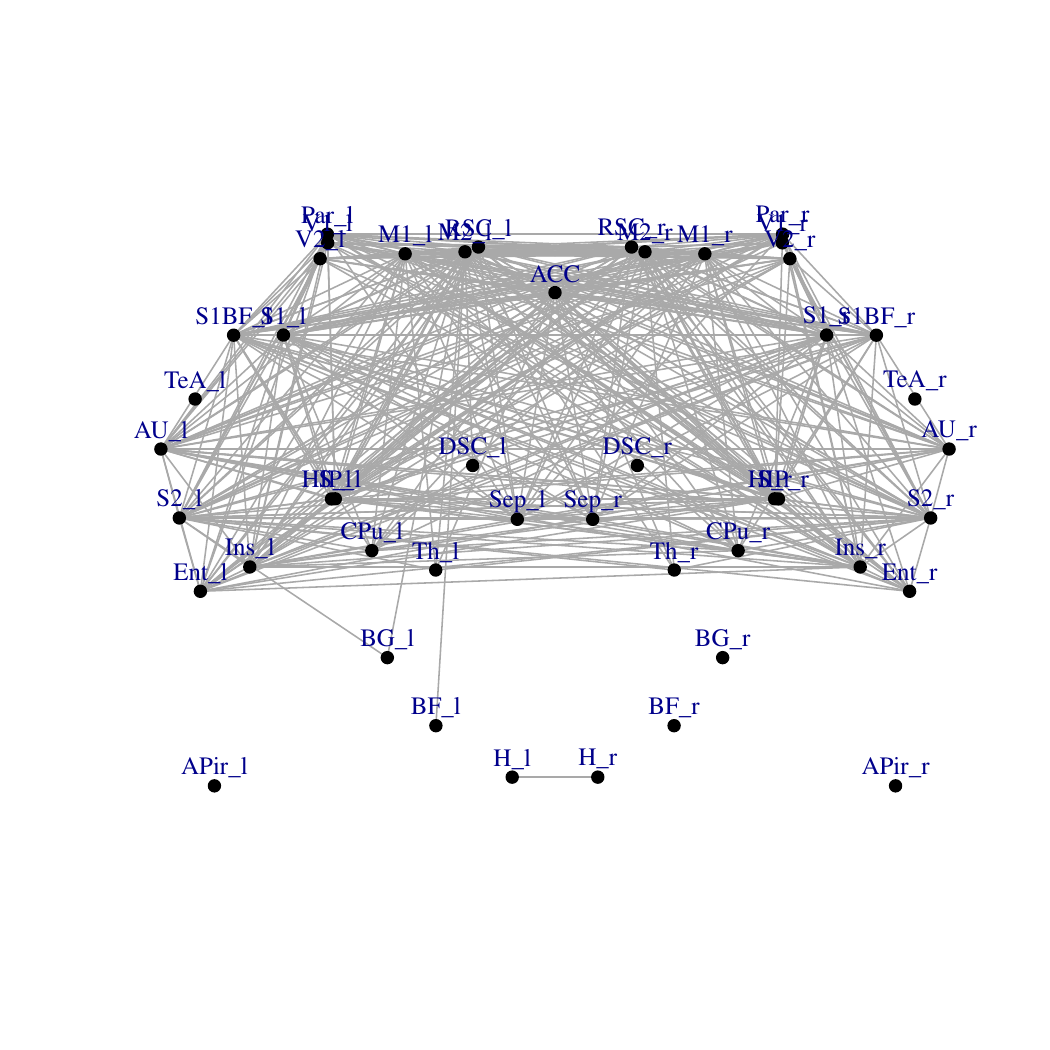}
            \caption{Live rat $20160614\_095825$}
        \end{subfigure}
        \caption{Brain functional connectivity network of a dead and two live rats (anesthetized with Isofluorane) inferred using our proposed correlation screening framework with the quantile-based threshold (CS + $\rho_{q,\alpha=0}^{a,b}$).}
        \label{fig:ratnetwork}
    \end{figure}
	}\fi

		\section{Discussion}
	\label{sec:limits}
	We have presented a novel approach to infer connectivity networks when nodes represent groups of correlated variables. We have formally established the importance of leveraging dependence structures to reliably discover inter-correlations. Our method consists in estimating, for each pair of groups, an inter-correlation distribution before deriving a tailored 
	threshold based on a correlation screening approach. In particular, we proposed simplified expressions for the mean number of discoveries that allow for easier theoretical and empirical manipulation, and flexibly take into account dependence within groups. Motivated by a real-world application, we have demonstrated the feasibility of our approach on a real dataset of rat brain images. 

 This work has several possible theoretical extensions.
First, while we provide a method with theoretical FPR control for any setting, we provide theoretical FWER control guarantees only under a weak dependence assumption. Nevertheless, this assumption is often unrealistic and relaxing it is difficult and would be an interesting direction to explore. Additionally, FWER approaches may sometimes be too conservative. On the other hand, the false discovery rate (FDR) enables to control the average number of FPs, which is often sufficient. A procedure to define correlation thresholds was proposed in \citep{Cai2016FDR} that leverages a quantity linked to the sum-based mean number of discoveries $E[N_e^{ab}]$. While they provide FDR control, it is only valid under some particular dependence conditions. It would nonetheless be interesting to extend their work to arbitrary dependence.
	
	In this paper, the aim was to reliably detect one edge at a time. It would then be interesting to build upon the proposed edge-centric correlation thresholds to develop a multiple testing framework so as to provide theoretical control over the estimation of the connectivity of all pairs of region, perhaps by leveraging existing bootstrapping techniques \citep{Cai2016FDR}.

	Finally, it would be valuable to provide practitioners with a way to quantify edge detection uncertainty. For instance, confidence intervals for each edge of the whole inferred network could be defined. Some work has been done to determine confidence intervals for correlation coefficients in the bivariate case, both for underlying normality, e.g., \citep{Ruben1966, muirhead2005aspects}, and unknown distributions \citep{Hu2020intervalCorr}. It would be worth exploring how these methods could build upon our approach to extend them to a more general case in order to account for dependence. 

\bmhead{Supplementary information} Proofs of the propositions are available in the appendix. Additional discussions and details can be found in the supplementary materials. Source code, including a notebook detailing how to reproduce the figures of this paper, is also
available at: \url{https://gitlab.inria.fr/q-func/csinference}.
    
    \bmhead{Acknowledgments}
    This work was supported by the project Q-FunC from Agence Nationale de
la Recherche under grant number ANR-20-NEUC-0003-02 and the National Science Foundation 
under grant IIS-2135859.
    
 \bibliography{biblio_paper.bib}

\begin{appendices}
    
\section{Proof of Proposition \ref{prop:wassdist}}
\label{append:wassinter}

    \subsection{U-scores} 
	\label{sec:prelim:uscores}
	
	Before proving Proposition \ref{prop:wassdist}, we need to introduce U-scores. \textit{U-scores} are an orthogonal projection of the Z-scores of random variables. They are confined to an $(n-2)$-sphere centered around 0 and with radius 1, denoted $S_{n-2}$, with $n$ the number of samples. We refer to \citep{hero2011} for a full definition.	U-scores namely provide a practical expression of the correlation coefficient as an inner product of U-scores: $R_{i,j}^{ab} = U_i^a{}^T U_j^b = 1 -  \| U_i^a - U_j^b \|^2 /2$, where  $U_i^a$, $U_j^b$ are the random variables of the U-scores of voxels $i$ and $j$ in regions $a$ and $b$, respectively, and $\|.\|^2$ is the squared Euclidean distance. Consequently, when U-scores are close to one another on $S_{n-2}$, they are associated with a high correlation.

		\subsection{Intuition behind Proposition \ref{prop:wassdist}.} Roughly speaking, Proposition \ref{prop:wassdist} means that if the Wasserstein distance between the densities of sample intra-correlation coefficients is sufficiently large, which means that the two distributions are highly different, then the average Euclidean distance between U-scores from the two regions is large too. Hence the average of inter-correlations is quite low. This phenomenon is illustrated in Figure \ref{fig:highdist} where two regions with different intra-correlation densities are depicted when $n=3$.
	
	\begin{figure}[!h]
		\centering
\includegraphics[width=.25\textwidth]{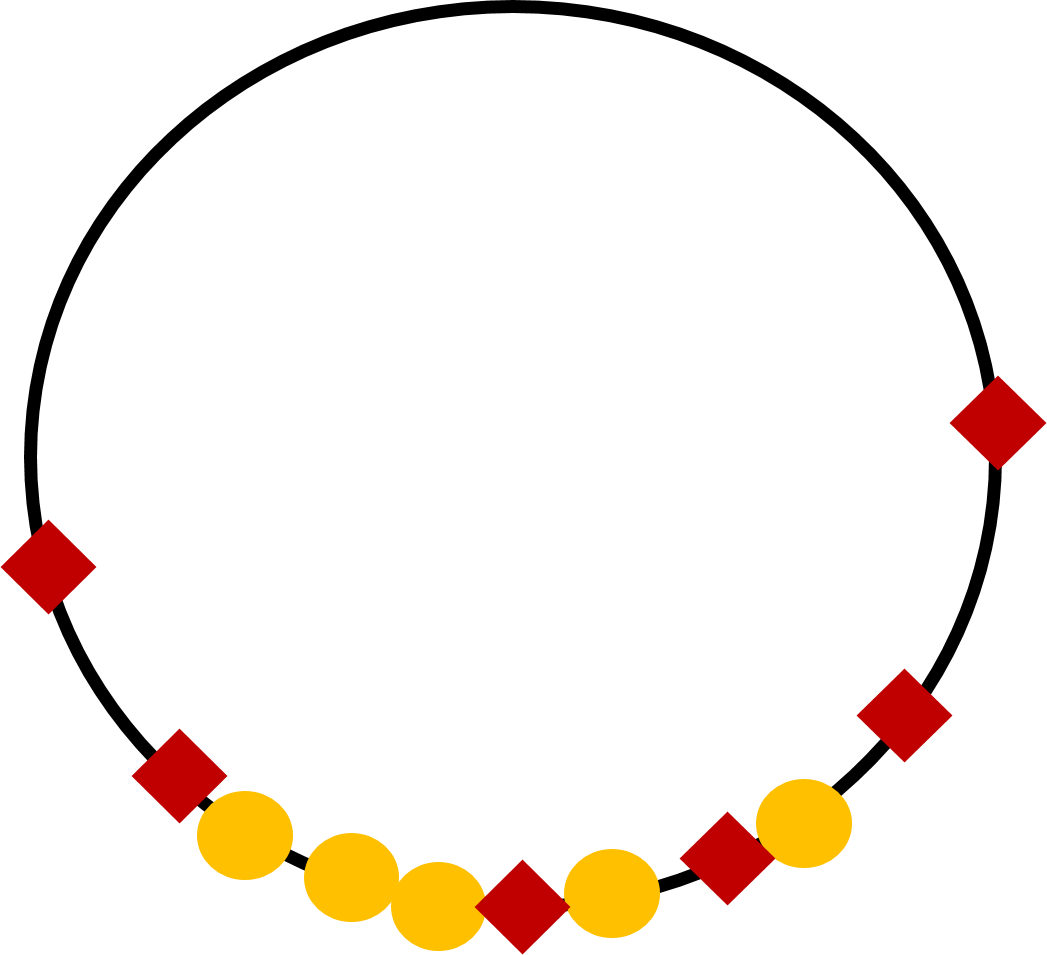}
		\caption{$S_{n-2}$ with $n=3$ and U-scores from two regions (red diamonds and orange discs) that have a high intra-correlation density Wasserstein distance. Recalling that a high Euclidean distance between U-scores implies a low correlation, we can intuitively observe that the average inter-correlation is upper-bounded.}
		\label{fig:highdist}
  \end{figure}
	\subsection{Proof of Proposition \ref{prop:wassdist}}
	Let us first remark:
		\begin{equation}
			\inf_{c\in [0,1]} \left( F_{R^{a,a}}^{-1}(c) - F_{R^{b,b}}^{-1}(c) \right)^2 \leq d_W^2(f_{R^{a,a}}, f_{R^{b,b}}) \leq \sup_{c \in [0,1]} \left( F_{R^{a,a}}^{-1}(c) - F_{R^{b,b}}^{-1}(c) \right)^2 
		\end{equation}
		
		and as $ \left( F_{R^{a,a}}^{-1} - F_{R^{b,b}}^{-1} \right)^2 $ is continuous on $[0,1]$, it attains its supremum and infimum. 
		
		Additionally, we can notice that for each $c \in [0,1]$, there exist two points $U,V \in S_{n-2}$ such that $F_{R^{a,a}}^{-1}(c) = 1 -  \| U - V \|^2 /2$, and similarly for region $b$. Moreover, for all $x,y \in \mathcal{R}_a$, there exists $c \in [0,1]$ such that, with their corresponding U-scores denoted $U_x,U_y$ (which are in $S_{n-2}$), $R^{a,a}_{x,y} = 1 -  \| U_x - U_y \|^2 /2 = F_{R^{a,a}}^{-1}(c) $, and analogously for region $b$. 
		
		Therefore, under the assumption $\min\limits_{c\in [0,1]} \left( F_{R^{a,a}}^{-1}(c) - F_{R^{b,b}}^{-1}(c) \right)^2  \geq A$, there exist $U_{x_a}, U_{y_a}, U_{x_b},U_{y_b} \in S_{n-2}$ such that 
		\begin{equation*}
			\min\limits_{c\in [0,1]} \left( F_{R^{a,a}}^{-1}(c) - F_{R^{b,b}}^{-1}(c) \right)^2 
			= \frac{1}{4}\left( \| U_{x_b} - U_{y_b} \|^2 - \| U_{x_a} - U_{y_a} \|^2 \right)^2,
		\end{equation*}
		and it follows for all $v_a,w_a \in \mathcal{R}_a$, $v_b,w_b \in \mathcal{R}_b$,
		\begin{equation*}
			A \leq \frac{1}{4} \left( \| U_{x_b} - U_{y_b} \|^2 - \| U_{x_a} - U_{y_a} \|^2 \right)^2 \leq \frac{1}{4} \left( \| U_{v_b} - U_{w_b} \|^2 - \| U_{v_a} - U_{w_a} \|^2 \right)^2.
		\end{equation*}

		Thus, expanding the term on the right and applying the triangle inequality, followed by the reverse triangle inequality,
		\begin{equation*}
		    \begin{split}
			2\sqrt{A} & \leq \left( \| U_{v_b} - U_{w_b} \| + \| U_{v_a} - U_{w_a} \| \right) \cdot  \Big\lvert \| U_{v_b} - U_{w_b} \| - \| U_{v_a} - U_{w_a} \| \Big\lvert \\
			& \leq \left( \| U_{v_b} - U_{v_a} \| + \| U_{v_a} - U_{w_b} \| + \| U_{v_a} - U_{w_b} \| + \| U_{w_b} - U_{w_a} \| \right)  \cdot \| U_{v_b} - U_{w_b} - (U_{v_a} - U_{w_a}) \| \\
			& \leq \left( \| U_{v_b} - U_{v_a} \| + \| U_{v_a} - U_{w_b} \| + \| U_{v_a} - U_{w_b} \| + \| U_{w_b} - U_{w_a} \| \right)  \cdot   \left( \| U_{v_b} - U_{v_a} \| + \|U_{w_b} - U_{w_a}\| \right) \\
			& \leq ({\| U_{v_b} - U_{v_a} \|^2} + {\| U_{v_b} - U_{v_a} \| \cdot\| U_{w_b} - U_{v_a} \|} + {\| U_{w_b} - U_{v_a} \| \cdot \|U_{w_b} - U_{w_a}\| }\\
			& + {\| U_{v_b} - U_{v_a} \| \cdot \| U_{w_b} - U_{w_a} \|} )\\
			& + ({\|U_{w_b} - U_{w_a}\|^2} + {\| U_{w_b} - U_{v_a} \| \cdot \| U_{v_b} - U_{v_a} \|} + { \|U_{w_b} - U_{w_a}\| \cdot \| U_{w_b} - U_{v_a} \|} \\
			& + {\| U_{v_b} - U_{v_a} \| \cdot \| U_{w_b} - U_{w_a} \|}).
			\end{split}
		\end{equation*}
		We can then notice
		\begin{alignat*}{2}
			\overline{\| U^{b} - U^{a} \|}^2 & = \left( \frac{1}{p_a p_b} \sum_{v_a \in \mathcal{R}_a} \sum_{v_b \in \mathcal{R}_b}  \| U_{v_b} - U_{v_a} \| \right) ^2 \\
			& = \frac{1}{(p_a p_b)^2} {\sum_{h_a \in \mathcal{R}_a} \sum_{h_b \in \mathcal{R}_b} \| U_{h_b} - U_{h_a} \|^2 }+ \\
			& \frac{1}{(p_a p_b)^2} {\sum_{h_a \in \mathcal{R}_a} \sum_{h_b \in \mathcal{R}_b} \sum_{k_b \in \mathcal{R}_b-\{h_b\}}  \| U_{h_b} - U_{h_a} \|\cdot  \| U_{k_b} - U_{h_a} \|} + \\
			& \frac{1}{(p_a p_b)^2} {\sum_{h_a \in \mathcal{R}_a} \sum_{h_b \in \mathcal{R}_b} \sum_{k_a \in \mathcal{R}_a-\{h_a\}} \| U_{h_b} - U_{h_a} \|\cdot  \| U_{h_b} - U_{k_a} \|} + \\
			& \frac{1}{(p_a p_b)^2} {\sum_{h_a \in \mathcal{R}_a} \sum_{h_b \in \mathcal{R}_b} \sum_{k_a \in \mathcal{R}_a-\{h_a\}} \sum_{k_b \in \mathcal{R}_b-\{h_b\}}  \| U_{h_b} - U_{h_a} \|\cdot  \| U_{k_b} - U_{k_a} \|}.
		\end{alignat*}
		Thus $\overline{\| U^{b} - U^{a} \|}^2 \geq \frac{1}{(p_a p_b)^2} \cdot \frac{(p_a p_b)^2}{2} \cdot 2\sqrt{A} = \sqrt{A}$. From the Cauchy-Schwarz inequality,
		$$ \overline{R^{a,b}} \leq 1 - \frac{ \overline{\| U^{b} - U^{a} \|}^2}{2},$$
		which completes the proof.

	\section{Proof of Proposition \ref{prop:NeNabineq}}
	\label{append:NeNabineq}
	
	Let us recall Proposition \ref{prop:NeNabineq}.
	\newtheorem*{Pineq}{Proposition~\ref{prop:NeNabineq}}
	\begin{Pineq}
			For all $\rho \in [0,1]$, 
		\begin{equation}
		    \widehat{\nu}_e^{ab}(\rho) 
      \leq \widehat{\nu}^{ab}(\rho). 
      \end{equation}
	\end{Pineq}
	\begin{proof}
		Since for all $\rho \in [0,1]$, $0 \leq \widehat{F}_{|R^{a,b}|}(\rho) \leq 1$, then $\widehat{F}_{|R^{a,b}|}(\rho)^{p_b} \leq \widehat{F}_{|R^{a,b}|}(\rho)$. 
        Thus, $\widehat{\nu}^{ab}(\rho) = 1- \widehat{F}_{|R^{a,b}|}(\rho)^{p_b} \geq 1 - \widehat{F}_{|R^{a,b}|}(\rho) = \widehat{\nu}_e^{ab} $
	\end{proof}
	
	\section{Additional insights on $\nu^{\MakeLowercase{ab}}$}
	\label{append:Nab}
	We can derive the following proposition.
	\begin{proposition}
		\label{prop:Nab}
		If, for a fixed $i = 1, \dots,p_a$, all sample inter-correlation coefficients $R_{i,j}^{a,b}$ are i.i.d., $\nu^{ab} = \Tilde{\nu}^{ab}$. 
	\end{proposition}
	
	\begin{proof}
		We can first remark $\max\limits_{j \in \mathcal{R}_b} \phi_{ij}^{ab} = 1 - \prod\limits_{j=1}^{p_b} (1-\phi_{ij}^{ab})$.
		Thus, 
		\begin{align*}
			E[N^{ab}] & = \sum\limits_{i=1}^{p_a} E[1 - \prod\limits_{j=1}^{p_b} (1-\phi_{ij}^{ab}) ] \\
			& = \sum\limits_{i=1}^{p_a} ( 1 - \prod_{j=1}^{p_b} E[(1-\phi_{ij}^{ab})] ) \text{ under the assumption of independence.}
		\end{align*}
		For a fixed $i = 1, \dots,p_a$, under the assumption all $| R^{
		a,b}_{i,j}|$ are identically distributed, then, for all $j,l = 1, \dots,p_b$ , $F_{| R^{a,b}_{i,j}|} = F_{| R^{a,b}_{i,l}|}$. Denote, $F_{| R^{a,b}|}$ the distribution function such that $F_{| R^{a,b}_{i,j}|} = F_{| R^{a,b}|}$ for all $i = 1, \dots, p_a $, $j = 1, \dots, p_b$. Thus,
		$ p_a \cdot \Tilde{\nu}^{ab} = E[N^{ab}] = p_a \cdot \left[1 - F_{|R^{a,b}|}^{p_b} \right] = p_a \cdot  \nu^{ab}. $
		
	\end{proof}

	\section{Further information about our implementation}
	Our implementation is based on R 4.2.3. All experiments were performed on a laptop running on Ubuntu 18.04 with eight 1.8GHz 64-bits Intel Core i7-10610U CPUs, 32 GB of memory and a 1 TB hard drive.   

\end{appendices}


\end{document}